\documentclass[12pt, draftcls, onecolumn, twoside]{IEEEtran}
\usepackage{amsmath,amssymb,amsfonts,amsthm}
\usepackage{algorithmic}
\usepackage{graphicx,pgfplots}
\usepackage{textcomp,multirow}
\usepackage{caption,subcaption}
\usepackage{bbm}

\newtheorem{theorem}{\textbf{Theorem}}

\newtheorem{definition}{\textit{Definition}}

\pgfplotsset{compat=1.17} 
\newcommand{\Zcal}{\mathcal{Z}}
\newcommand{\Kappa}{\mathcal{K}}
\newcommand{\Wcal}{\mathcal{W}}
\newcommand{\Qcal}{\mathcal{Q}}
\newcommand{\Ebb}{\mathbb{E}}
\newcommand{\Nbb}{\mathbb{N}}

\newcommand{\Vbf}{\mathbf{V}}
\newcommand{\Pbf}{\mathbf{P}}
\newcommand{\Qbf}{\mathbf{Q}}
\newcommand{\qbf}{\mathbf{q}}
\newcommand{\Xbf}{\mathbf{X}}
\newcommand{\dbf}{\mathbf{d}}
\newcommand{\abf}{\mathbf{a}}
\newcommand{\kpr}{k^{\prime}}
\newcommand{\fpr}{f^{\prime}}
\newcommand{\bpr}{b^{\prime}}

\newcommand{\Prob}[1]{\mathbb{P}\Big\{ #1 \Big\}}
\newcommand{\RPIR}[2]{\bigg( 1 + \frac{1}{#1} + \cdots + \frac{1}{#1^{#2 - 1}}\bigg)}
\newcommand{\RPIRs}[2]{\Bigg( \frac{1}{#1} + \frac{1}{{#1}^{2}} \cdots + \frac{1}{#1^{#2}}\Bigg)}
\newcommand{\RPIRg}[2]{\Bigg( 1 + \frac{1}{#1} + \frac{1}{{#1}^{2}} \cdots + \frac{1}{#1^{#2}}\Bigg)}
\newcommand{\RPIRgshort}[2]{\Bigg( 1 + \cdots + \frac{1}{#1^{#2}}\Bigg)}
\newcommand{\RPIRshort}[2]{\bigg( 1 + \cdots + \frac{1}{#1^{#2 - 1}}\bigg)}
\mathchardef\mhyphen="2D
\title{Cache-Aided Multi-User Private Information Retrieval using PDAs}
\author{Kanishak Vaidya, B Sundar Rajan \\
Department of Electrical Communication Engineering, IISc Bangalore, India \\
E-mail: \{kanishakv, bsrajan\}@iisc.ac.in \vspace{-3mm}
}

\begin{document}
\maketitle

\begin{abstract}
  We consider the problem of cache-aided multi-user private information retrieval (MuPIR). In this problem, $N$ independent files are replicated across $S \geq 2$ non-colluding servers. There are $K$ users, each equipped with cache memory which can store $M$ files. Each user wants to retrieve a file from the servers, but the users don't want any of the servers to get any information about their demand. The user caches are filled with some arbitrary function of the files before the users decide their demands, known as the placement phase. After deciding their demands, users cooperatively send queries to the servers to retrieve their desired files privately. Upon receiving the queries, servers broadcast coded transmissions which are a function of the queries they received and the files, known as the delivery phase. Conveying queries to the servers incurs an upload cost for the users, and downloading the answers broadcasted by the servers incurs a download cost. To implement cache-aided MuPIR schemes, each file has to be split into $F$ packets. In this paper, we propose MuPIR schemes that utilize placement delivery arrays (PDAs) to characterize placement and delivery. Proposed MuPIR schemes significantly reduce subpacketization levels while slightly increasing the download cost. The proposed scheme also substantially reduces the upload cost for the users. For PDAs based on {\it Ali-Niesen} scheme for centralized coded caching, we show that our scheme is order optimal in terms of download cost. We recover the optimal single-user PIR scheme presented by {\it Tian et al.} in [``Capacity-Achieving Private Information Retrieval Codes With Optimal Message Size and Upload Cost,'' IEEE Trans. Inform. Theory, 2019] as a special case. Our scheme also achieves optimal rate for single-user cache-aided PIR setup reported in [R. Tondon, ``The capacity of cache aided private information retrieval,'' Annual Allerton Conference on Communication, Control and Computing, 2017.] \footnote{Part of the content of this manuscript has been communicated to IEEE Information Theory Workshop (ITW) 2023, to be held in Saint-Malo, France from 23-28 April 2023 \cite{VaR2}.}
\end{abstract}

\section{Introduction}\label{sec:Introduction}
The problem of Private Information Retrieval (PIR), first described in~\cite{Chor95PIR}, seeks efficient ways for a user to retrieve data from distributed the servers privately. A user wishes to retrieve one file among a set of files stored across servers. But the servers should not know the identity of the desired file. A PIR scheme that minimizes the download cost for the user is described in~\cite{Sun17PIR}. After that, the PIR problem is solved for various other settings, for instance~\cite{Sun18cPIR,Lin19wPIR,Chen20siPIR}.

Currently, PIR is being studied with another content delivery scenario called coded caching. As described in~\cite{MAli13CodedCaching}, in coded caching, there are multiple users equipped with user cache and one server storing some files. During off-peak hours, users fill their caches and then, during peak network traffic hours, demand files from the server. The server will perform coded transmissions such that a single transmission can benefit multiple users simultaneously. After receiving the transmissions, users will be able to decode their demanded files with the help of content stored in their cache. Recently, in~\cite{Ming21CaMuPIR}, a cache-aided PIR strategy is described where multiple users, each having access to dedicated caches, want to recover files from non-colluding servers privately. An order optimal strategy is described that combines coding benefits of PIR in~\cite{Sun17PIR} and coded caching~\cite{MAli13CodedCaching}.

\subsubsection*{Notations}\label{sec:subsub:Notations}
For integers $m$ and $n$, $[m:n]$ is set of integers $N,$ $\{ m\leq N \leq n\}$. $\Nbb$ is the set of all positive integers, $\Nbb = \{1 , 2 , 3 , \cdots \}$. $[N]$ is same as $[1:N]$. For a set $\mathcal{S}$ of size $|\mathcal{S}|$ and integer $N \leq |\mathcal{S}|$, $\binom{\mathcal{S}}{N}$ denote set of all subsets of $\mathcal{S}$ of size $N$. For set $\{a_n | n \in [N]\}$ and $\mathcal{N} \subseteq [N]$, $a_{\mathcal{N}}$ denotes set $\{a_n | n \in \mathcal{N}\}$. Given a vector $\Vbf = (V_{0} , V_{1} , \ldots , V_{N})$ and integers $m$, $n$ such that $0 \leq m \leq n \leq N$, $\Vbf(m:n) \triangleq (V_{m} , \ldots V_{n})$. For integers $M$ and $N$, ${(M)}_{N}$ is $M \bmod N$. 

\subsection{Coded Caching}\label{sec:sub:CodedCaching_intro}
In a centralized coded caching system described in~\cite{MAli13CodedCaching}, a server stores $N$ independent files $\{W_{0} , \ldots , W_{N-1}\}$, each of size $L$ bits. There are $K$ users, each equipped with a cache memory of $ML$ bits. The system works in two phases. In \textit{delivery phase}, when the network is not congested, the server fills the caches with the contents of the files. Then, in \textit{delivery phase}, all users wish to retrieve some files from the servers. User $k$ wishes to retrieve file $d_{k} \in [0:N-1]$. Every user conveys the index of their desired file to the server. After receiving the demands from the users, the server broadcasts coded transmissions $\Xbf$ of size $R_{cc}L$ bits. The transmission $\Xbf$ is a function of the files stored at the server and users' demand. After receiving the coded transmission $\Xbf$, all the users should be able to retrieve their desired files with the help of their cache content. The quantity $R_{cc}$ is defined as the rate of the coded caching system, and it measures the size of the server's transmissions.

In~\cite{MAli13CodedCaching} a placement and delivery scheme was provided, known as the MAN scheme, for $M = tN/K$ for some $t \in [0:K]$ that achieves rate (as a function of $M$)
\begin{equation}\label{eq:MAN_rate}
	R_{cc}(M) = K \Big(1-\frac{M}{N}\Big)\frac{1}{1+\frac{KM}{N}}.
\end{equation}
For other memory points, lower convex envelope of points ${\Big(tN/K , R_{cc}(tN/K)\Big)}_{t \in [0:K]}$ can be achieved using memory sharing. In this scheme, each file has to be divided into $\binom{K}{t}$ subfiles, which is known as the \textit{subpacketization level}. This high subpacketization level is a major drawback of the MAN scheme.

Let $R^{*}_{cc}(M)$ be the minimum achievable rate for a coded caching problem with $K$ users, $N$ files and cache size of $M$. It is shown in~\cite{Ramamoorthy17LowerBoundCC}, that rate achieved by the MAN scheme is optimal within a multiplicative factor of $4$, i.e.
\begin{equation}\label{eq:CCoptimality}
  R_{cc}^{*}(M) \geq \frac{R_{cc}(M)}{4}.
\end{equation}
It is also shown in~\cite{YMA18RMTCachingUncoded} that for worst case demands, the rate $R_{cc}(M)$ is optimal for the case $N \geq K$, considering uncoded placement. Considering $R^{*}_{ucc}(M)$ to be the optimal rate for a coded caching problem with uncoded placement, we have
\begin{equation}\label{eq:unCCoptimality}
  R_{cc}(M) = R^{*}_{ucc}(M).
\end{equation}

\subsection{Placement Delivery Arrays}\label{sec:sub:PDA_intro}
To characterize coded caching schemes, the concept of \textit{placement delivery array} (PDA) was introduced in~\cite{QYan17PDA}. PDAs can be used to describe the placement and the delivery phase using a single array. A PDA is defined as follows:
\begin{definition}
For  positive integers $K,F,Z$ and $S$, an $F\times K$ array  $\mathbf{P}=[p_{f,k}]$, $f\in [F], k\in[K]$, composed of a specific symbol $``*"$  and $S$ non-negative integers $1,\cdots, S$, is called a $(K,F,Z,S)$ placement delivery array (PDA) if it satisfies the following conditions:
\begin{enumerate}
  \item [C$1$.] The symbol $``*"$ appears $Z$ times in each column;
  \item [C$2$.] Each integer occurs at least once in the array;
  \item [C$3$.] For any two distinct entries $p_{f_1,k_1}$ and $p_{f_2,k_2}$,    $p_{f_1,k_1}=p_{f_2,k_2}=s$ is an integer only if
  \begin{enumerate}
     \item [a.] $f_1\ne f_2$, $k_1\ne k_2$, i.e., they lie in distinct rows and distinct columns; and
     \item [b.] $p_{f_1,k_2}=p_{f_2,k_1}=*$, i.e., the corresponding $2\times 2$ sub-array formed by rows $f_1,f_2$ and columns $k_1,k_2$ must be of the following form
  \begin{eqnarray*}
    \left[\begin{array}{cc}
      s & *\\
      * & s
    \end{array}\right]~\textrm{or}~
    \left[\begin{array}{cc}
      * & s\\
      s & *
    \end{array}\right]
  \end{eqnarray*}
   \end{enumerate}
\end{enumerate}
\end{definition}\hfill \qedsymbol%

A coded caching scheme characterized by a $(K,F,Z,S)$ PDA has cache size $M/N = Z/F$ and achieves rate $R_{cc} = S/F$ and the subpacketization level is $F$. For a given $(K,F,Z,S)$ PDA $\Pbf$, we define $\Kappa_{s} \triangleq \{k \in [K] : p_{f,k} = s, \mbox{ for some } f \in [F]\}$, i.e.\ $\Kappa_{s}$ is the set of indices of columns of $\Pbf$ that has the integer $s$. A $(K,F,Z,S)$ PDA is called a $g$-regular PDA if $|\Kappa_{s}| = g, \forall s \in [S]$ for some $g \in \Nbb$. Such PDAs are denoted by $g\mhyphen(K,F,Z,S)$.

In~\cite{QYan17PDA}, two PDA constructions were presented. The placement and delivery schemes corresponding to these PDAs have a similar rate as the rate in the MAN scheme, but there is a significant improvement in the subpacketization level. In~\cite[Theorem 4]{QYan17PDA}, for any given $q, m \in \Nbb^{+}, q \geq 2$ a $(m+1)\mhyphen\big( q(m+1),q^{m},q^{m-1},q^{m+1}-q^{m} \big)$ PDA construction is provided. For this PDA $M/N = 1/q$, and rate is $R = q-1$. In~\cite[Theorem 5]{QYan17PDA}, for any given $q, m \in \Nbb^{+}, q \geq 2$ a $(q-1)(m+1)\mhyphen\big(q(m+1),(q-1)q^{m}, (q-1)^{2}q^{m-1},q^{m}\big)$ PDA construction is provided with $M/N = 1 - (1/q)$, and rate is $R = 1 / (q-1)$.

\subsection{Private Information Retrieval}\label{sec:sub:PIR_intro}
In Private Information Retrieval (PIR) there is one user and a set of $N$ independent files $\mathcal{W} = {\{W_n\}}_{n = 0}^{N-1}$ replicated across $B$ non-colluding servers indexed by integers in the set $[0:B-1]$. A user wants to retrieve one out of $N$ files, say file $W_{D} , {D} \in [0:N-1]$, but doesn't want the servers to know the identity of the file. In other words, the user wants to hide the index ${D}$ from the servers. In order to retrieve this desired file privately, the user generates $B$ queries ${\{Q_b^{D}\}}_{b = 1}^B$ and sends the  query $Q_b^{D}$ to server $b$. After receiving their respective queries, servers will construct answers which are a function of the query they got and the files they have. Server $b$ will construct answer $A_b^{D}(\mathcal{W})$ and send it to the user. After receiving answers from all the $B$ servers, the user should be able to decode the desired file. Privacy and correctness conditions are formally stated as follows:

For privacy we need that 
\[
  I({D} ; Q_b^{D}) = 0 , \forall b \in \{1 , \hdots , B\},
\]
and for correctness
\[
  H(W_{D} | {D} , A_{1}^{D}(\mathcal{W}) \hdots A_{B}^{D}(\mathcal{W}) , Q_{1}^{D} \hdots Q_{B}^{D}) = 0.
\]
The rate of PIR is a parameter that describes the download cost for the user (or the transmission cost for the servers), which is defined as
\[
  R_{PIR} = \frac{\sum_{b = 1}^{B}(H(A_b^{D}(\mathcal{W})))}{H(W_{D})}.
\]
A rate optimal scheme is provided in~\cite{Sun17PIR} with optimal rate $R^*_{PIR}$ given by
\begin{equation}\label{eq:capacity_PIR}
  R^*_{PIR} = \RPIR{B}{N}.
\end{equation}
The retrieval scheme provided in~\cite{Sun17PIR} requires dividing each file into $B^{N}$ subfiles and incurs an upload cost of $BN \log_{2} \Big( \frac{B^{N}!}{B^{N-1}!} \Big)$ bits. Another capacity achieving PIR scheme is provided in~\cite{Tian19optimalUpSubPIR}, which requires dividing each file into $B-1$ subfiles and incurs upload cost of only $B(N-1) \log_{2} B$ bits. This upload cost and subpacketization level are shown to be optimal among all capacity-achieving linear PIR codes in~\cite{Tian19optimalUpSubPIR}.

\subsection{Cache-Aided Private Information Retrieval}\label{sec:sub:CAPIR_intro}
For a single user cache-aided PIR setup, with $N$ files of size $L$ bits each, and $B$ non-colluding servers, where a single user can store $ML$ bits, the optimal rate was shown in~\cite{Tandon19CaSuPIR} to be
\begin{equation}\label{eq:capacity_CaPIR}
	R_{CaPIR} = \Big(1 - \frac{M}{N}\Big)\RPIR{B}{N}.
\end{equation}
The above rate can be achieved by storing $M/N$ fraction of every file in the cache of the user and then querying the remaining $\Big( 1 - \frac{M}{N} \Big)$ fraction using the PIR scheme explained in~\cite{Sun17PIR}.

For multiple cache-aided users, various PIR strategies were given in~\cite{Ming21CaMuPIR}. For the cache-aided MuPIR problem with $N=2$ files, $K = 2$ users and $B \geq 2$ servers, a novel {\it cache-aided interference alignment} (CIA) based achievable scheme was given that achieves the rate~\cite[Theorem 1]{Ming21CaMuPIR}:
\begin{align*}
  R_{\textrm {CIA}}(M) = \begin{cases}
                           2(1-M),& 0\le M\le \frac {B-1}{2B} \\
                           \frac {(B+1)\left ({3-2M }\right)}{2B+1}, & \frac {B-1}{2B}\le M \le \frac {2(B-1)}{2B-1} \\
                           \left ({1-\frac {M}{2}}\right)\left ({1+\frac {1}{B}}\right), & \frac {2(B-1)}{2B-1} \le M\le 2.
                         \end{cases}
\end{align*}
This rate is optimal for $B = 2,3$ servers. Moreover the rate $R_{\textrm{CIA}}(M)$ is optimal when $M\in \left[{ 0, \frac {B-1^{\vphantom {)}} }{2B_{\vphantom {)}} } }\right]\cup \left[{ \frac {2(B-1)}{2B-1},2}\right]$.

In~\cite[Theorem 2]{Ming21CaMuPIR}, an achievable scheme called the \textit{product design} was proposed for general $N$, $K$ and $B$ that achieves rate
\begin{align*}
  R_{\mathrm{ PD}}(M) = \min \big \{N-M,\widehat {R}(M)\big \}, \mbox{where }  
  \widehat {R}(M) = \frac {K-t}{t+1} \RPIR{B}{N}
\end{align*}
and $t = \frac{KM}{N} \in [0:K]$. This rate was shown to be order optimal within a multiplicative factor of $8$. In this product design, the subpacketization level is $B^{N} \binom{K}{t}$ and the upload cost is $(t+1)\binom{K}{t+1}BN\log_{2}(\frac{B^{N}!}{B^{N-1}!})$ bits.

\subsection{Our Contributions}\label{sec:sub:our_contrib_intro}
In this paper, we use PDAs to construct PIR schemes for cache-aided MuPIR setups. For any given $(K,F,Z,S)$ PDA, we propose a PIR scheme for a cache-aided MuPIR setup with $K$ users each equipped with a cache of size $ZNL/F$ bits, $B \geq 2$ non-colluding servers and $N$ files. The proposed scheme incurs low upload cost, and the subpacketization level is also small. The main contributions of this paper are listed as follows:
\begin{itemize}
	\item Theorem~\ref{th:achievability} states the rate, subpacketization level and upload cost incurred by our proposed achievable scheme for any given PDA.
  \item The single user PIR scheme with optimal rate, subpacketization level and upload cost provided in~\cite{Tian19optimalUpSubPIR} is recovered as a special case of Theorem~\ref{th:achievability}.
  \item The rate provided in Theorem~\ref{th:achievability} is equal to the optimal rate for the special case of cache-aided single-user PIR setup described in~\cite{Tandon19CaSuPIR}.
  \item When the proposed scheme is specialized to the PDA corresponding to the MAN scheme, the resulting MuPIR scheme is shown to be order optimal in Theorem~\ref{th:order_optimality}.
  \item The proposed scheme corresponding to the PDAs given in~\cite{QYan17PDA} is compared to the MuPIR scheme given in~\cite{Ming21CaMuPIR}. The rate in our scheme is higher, but there is a significant reduction in the subpacketization level and the upload cost.
\end{itemize}

\section{System Model}\label{sec:system_model}
\begin{figure}
  \centering
  \includegraphics[width = 0.65\textwidth]{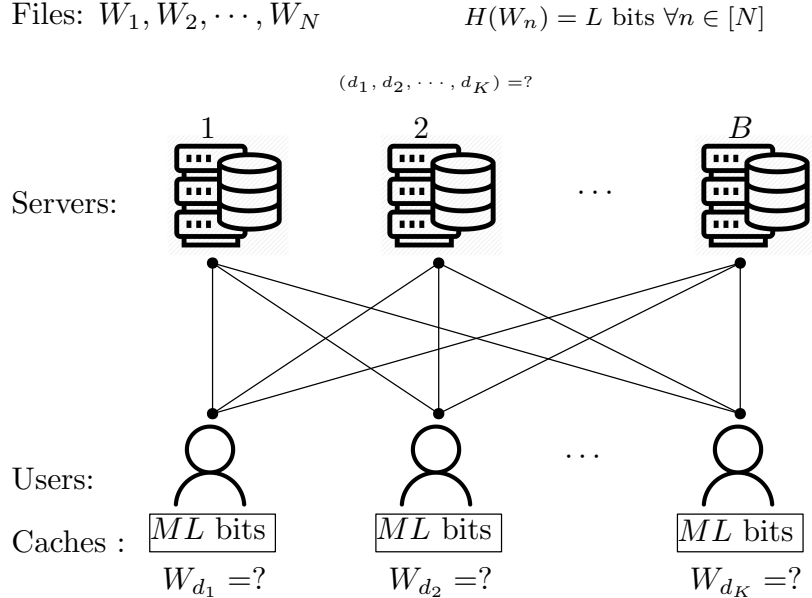}
  \caption{Dedicated cache-aided multi user PIR system.}
  \label{fig:system_model}
\end{figure}

There are $B$ non-colluding servers and a set of $N$ files $\Wcal = \{W_1 , W_2 , \cdots , W_N\}$ that are replicated across all the servers. Every file is of $L$ bits, i.e.\
\[
  H(W_k) = L \mbox{ bits } \forall k \in [K].
\]
We consider cache-aided systems derived from PDAs. Consider a $(K , F , Z , S)$ PDA as described in Section~\ref{sec:sub:PDA_intro}. It corresponds to a cache-aided system with $K$ users, each equipped with a dedicated cache of size $ML$ bits where $\frac{M}{N} = \frac{Z}{F}$. The caching system works in two phases.

\subsubsection*{Placement Phase} In this phase, servers fill the caches with the content of the files. Let $\Zcal_k$ be the content stored in the cache of user $k$. Then
\[
  H(\Zcal_k) \leq ML = \frac{Z}{F} NL \mbox{ bits } \forall k \in [K].
\]
Note that in the placement phase, servers don't know the future demands of the users, and the placement is done without the knowledge of what the users are going to demand in the next phase.
\subsubsection*{Private Delivery Phase} In the delivery phase, each user wishes to retrieve a file from the servers. Let user $k$ want the file $W_{d_k} , \forall k \in [K]$. We define $\dbf = (d_1 , d_2 , \cdots , d_K)$ to be the demand vector. Also, users don't want any of the servers to get any information about the demand vector. To achieve this, users cooperatively construct $B$ queries, $\{Q_b^{\dbf}\} , \forall b \in [B]$, one for each server. Query $Q_b^{\dbf}$ is sent to server $b$, and after getting this query, server $b$ responds with answer $A_b^{\dbf}(\Wcal , Q_b^{\dbf})$ of size $R_{b}L$ bits, which is a function of the queries sent to the individual server and the files stored across these servers. $A_b^{\dbf}(\Wcal , Q_b^{\dbf})$ is broadcasted to all the users via an error free link. After receiving the answers transmitted by the servers, each user should be able to decode its desired file using these answers and the cache contents respective user have access to. Formally we can say that
\begin{align}
  \mbox{For privacy:} ~~~ & I(\dbf ; Q_b^{\dbf} , \Wcal) = 0 , \forall b \in [B]. \\
  \mbox{For Correctness:} ~~~ &  H(W_{d_k} | Q_{[B]}^{\dbf} , A_{[B]}^{\dbf} , \Zcal_k , \dbf) = 0 , \forall k \in [K].
\end{align}

The following parameters characterize the performance of a cache-aided multi-user PIR:
\subsubsection*{Rate} This measures the average number of bits broadcasted by the servers to the users. We define the rate to be
\begin{equation}\label{eq:rate_definition}
  R = \Ebb\Bigg\{\sum_{b = 0}^{B-1} R_{b}\Bigg\}
\end{equation}
where $R_{b}L$ is the size of the answer $A_{b}^{\dbf}(\Wcal , Q_{b}^{\dbf})$ broadcasted by server $b$. This quantity should be minimized as the average size of broadcast performed by the servers is $RL$ bits.
\subsubsection*{Subpacketization level} This is the number of subfiles a file has to be divided into during the placement and delivery phase. The subpacketization level also has to be kept small.
\subsubsection*{Upload cost} This is the total number of bits users must send to the servers to convey the queries. The upload cost is
\begin{equation}
  U = \sum_{b \in [B]} H(Q_b^{\dbf}).
\end{equation}
We also aim to keep the upload cost as low as possible.

\section{Main Results}\label{sec:main_results}
In this section, we present the main results of this paper and compare the results with previous results on cache-aided MuPIR.
\begin{theorem}\label{th:achievability}
  Given a $(K , F , Z , S)$ PDA, there exist a cache-aided multi-user PIR scheme with $B$ non-colluding servers, $N$ files, $K$ users each equipped with a dedicated cache of size $M / N = Z / F$ with the following parameters
  \begin{itemize}
    \item Rate: $\min\{N-M , R\}$ where
          \[
            R = \frac{S}{F} \Bigg( 1 + \frac{1}{S}\sum_{s \in [S]} \RPIRs{B}{|\Kappa_{s}| (N-1)} \Bigg).
          \]
          \item Subpacketization level: $(B - 1)F$
          \item Upload Cost: $BK(N-1)\log_{2}B$ bits.
    \end{itemize}
\end{theorem}
\begin{proof}
  Rate $N-M$ can be achieved by storing same $M/N$ fraction of each file in the cache of every user. Then in the delivery phase, the remaining $(1 - M/N)$ fraction of each file can be transmitted to the users. An achievable scheme is provided in Section~\ref{sec:sub:general-scheme} that privately retrieves $K$ files with rest of the above-mentioned parameters.
\end{proof}

\textbf{Remark:} For a $g$ regular PDA, for which $|\Kappa_{s}| = g , \forall s \in [S]$, the rate is given by
\begin{align}
  R = \frac{S}{F} \RPIRg{B}{g(N-1)}.
\end{align}

\textbf{Remark:} For the special case of only one user, i.e.\ $K=1$, the PDA has only one column containing $Z$ ``*''s and $S$ distinct integers. Therefore,
\[
  S = F - Z, \implies \frac{S}{F} = 1 - \frac{Z}{F} = 1 - \frac{M}{N}.
\]
For such a PDA, $|\Kappa_{s}| = 1, \forall s \in [S]$. In this case, the rate achieved is given by
\begin{align*}
  R &= \frac{S}{F} \Bigg( 1 + \frac{1}{S}\sum_{s \in [S]} \RPIRs{B}{(N-1)} \Bigg) \\
  &= \frac{S}{F}\RPIR{B}{N} \\
  &= \Big(1 - \frac{M}{N}\Big)\RPIR{B}{N}
\end{align*}
which is equal to the optimal rate of cache-aided single-user PIR setup described in~\cite{Tandon19CaSuPIR}.

\textbf{Remark:} If $K=1$ and $M=0$ (i.e.\ single-user no cache setup), a trivial $(1,1,0,1)$ PDA can be used. With this PDA, our scheme achieves optimal subpacketization level and optimal upload cost as described in~\cite{Tian19optimalUpSubPIR} alongwith optimal PIR rate described in~\cite{Sun17PIR,Tian19optimalUpSubPIR}.

\subsection{Example: MuPIR based on a PDA}\label{sec:sub:irregularPDA}
Here, we will demonstrate the results of Theorem~\ref{th:achievability} for the following PDA which is not a regular PDA.
\begin{equation}
	\Pbf = \begin{bmatrix}
    * & * & * &  * &  * &  1 &  2 &  4 \\
    * & * & * &  1 &  2 &  * &  * &  5 \\
    * & * & * &  4 &  5 &  7 &  8 &  * \\
    1 & 2 & 3 &  * &  * &  * &  * & 10 \\
    4 & 5 & 6 &  * &  * & 10 & 11 &  * \\
    7 & 8 & 9 & 10 & 11 &  * &  * &  * \\
  \end{bmatrix}
\end{equation}
This is a $(8 , 6 , 3 , 11)$ PDA with
\begin{align*}
	\Kappa_{1} &= \{1 , 4 , 6\} & \Kappa_{5} &= \{2 , 5 , 8\}  & \Kappa_{ 9} &= \{3\}          \\
	\Kappa_{2} &= \{2 , 5 , 7\} & \Kappa_{6} &= \{3\}          & \Kappa_{10} &= \{4 , 6 , 8\}  \\
	\Kappa_{3} &= \{3\}         & \Kappa_{7} &= \{1 , 6\}      & \Kappa_{11} &= \{5 , 7\}      \\
	\Kappa_{4} &= \{1 , 4 , 8\} & \Kappa_{8} &= \{2 , 7\}      &             &                 \\
\end{align*}
and therefore,
\begin{align*}
	|\Kappa_{1}| &= 3 & |\Kappa_{5}| &= 3 & |\Kappa_{ 9}| &= 1  \\
	|\Kappa_{2}| &= 3 & |\Kappa_{6}| &= 1 & |\Kappa_{10}| &= 3  \\
	|\Kappa_{3}| &= 1 & |\Kappa_{7}| &= 2 & |\Kappa_{11}| &= 2. \\
	|\Kappa_{4}| &= 3 & |\Kappa_{8}| &= 2 &               &     \\
\end{align*}
For this PDA, we consider a MuPIR setup with $K = 8$ users, each equipped with a cache of size $M/N = 0.5$, $B = 2$ servers and $N = 8$ files. Then according to Theorem~\ref{th:achievability}, the rate achieved by the MuPIR scheme based on this PDA is given by
\begin{align*}
	\frac{11}{6}\left( 1 + \frac{1}{11} \left( 5 \left( \frac{1}{2} + \cdots + \frac{1}{2^{7 \times 3}} \right) + 3\left( \frac{1}{2} + \cdots + \frac{1}{2^{7 \times 2}}\right) + 3\left( \frac{1}{2} + \cdots + \frac{1}{2^{7}}\right) \right) \right) \approx 3.663.
\end{align*}
The subpacketization level is $(B-1)F = 6$ subfiles (which is the same as the subpacketization level of the coded caching setting without PIR constraints), and the upload cost is $112$ bits.

\subsection{Scheme based on MAN PDA}\label{sec:sub:MAN-PDA}
Consider the PDA based on the Maddah-Ali and Niesen scheme~\cite{MAli13CodedCaching}. The PDA corresponding to the MAN scheme is a $(t+1)\mhyphen\Big(K , \binom{K}{t} , \binom{K-1}{t-1} , \binom{K}{t+1}\Big)$ PDA for $t = KM/N \in \Nbb$. In TABLE~\ref{tab:MN-PDAvsPD} we compare the parameters of the MuPIR scheme using the MAN-PDA with the product design. We can see that the rate achieved in our scheme is higher than the rate achieved by the product design, but this results in a significant reduction in the upload cost in our proposed scheme. Consider dedicated cache setup with $B = 2$ servers, $K = 4$ users and $N = 4$ files. In Figure~\ref{fig:man-pda-vs-pd-rate} we plot the rate achieved by the product design in this setting and the rate achieved by the MAN scheme based PDAs. We can see that the rate achieved by our scheme is higher than the rate achieved by the product design. Although the rate is higher, the subpacketization level for PDA based scheme would be $\binom{4}{M}$ for $M \in [4]$ whereas the subpacketization level is $16\binom{4}{M}$ for the product design. The subpacketization level is $16$ times higher in the product design as compared to PDA based scheme. Also, in terms of upload cost, PDA based scheme is performing better than the product design. In Figure~\ref{fig:man-pda-vs-pd-upload}, we plot the ratio of the upload cost of product design and the PDA-based scheme. For $M/N = 1$ upload cost for all coded caching schemes is zero. For other values, i.e. for $M/N \in [0 , \frac{1}{4} , \frac{1}{2} , \frac{3}{4}]$ we can see that the upload cost in product design is higher than the upload cost for the PDA based schemes.

\begin{figure}
  \centering
	\includegraphics[width = 0.9\textwidth]{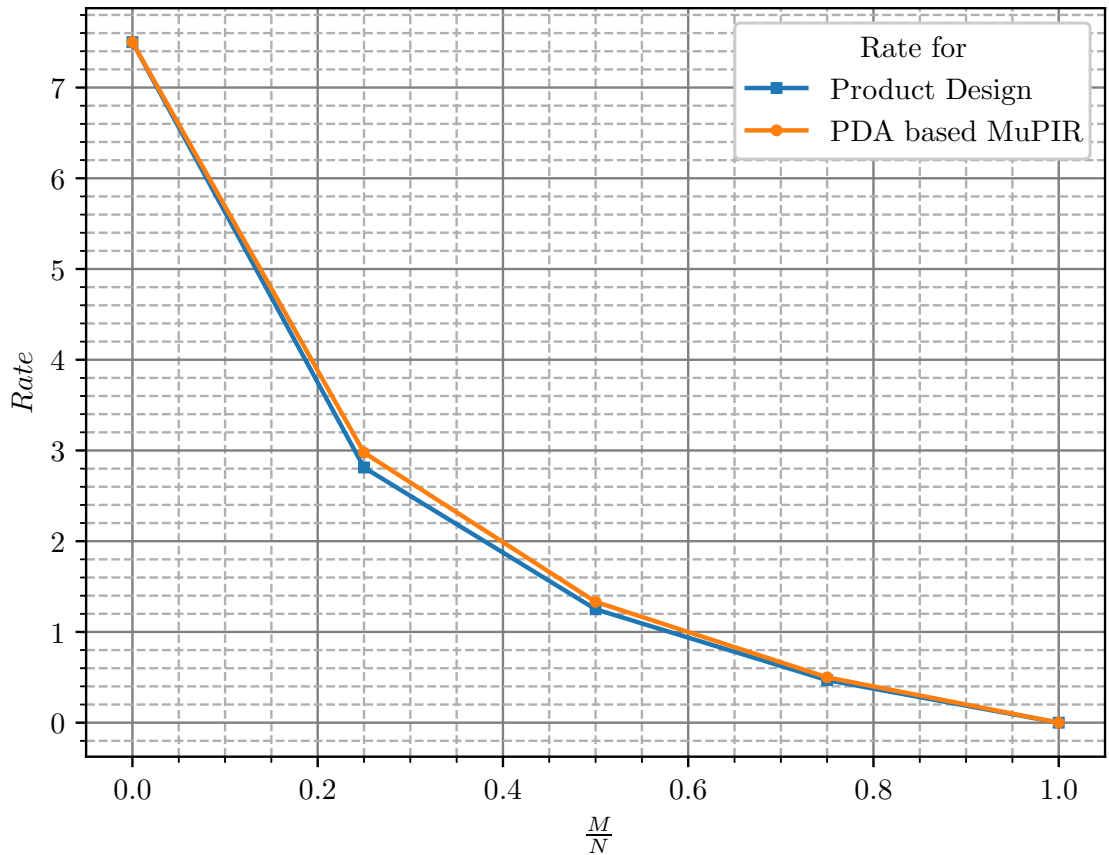}
  \caption{Comparing rates for Product design and MAN-PDA based MuPIR strategy with $B=2$ servers, $N=4$ files and $K = 4$ users for various values of $M/N$.}
  \label{fig:man-pda-vs-pd-rate}
\end{figure}

\begin{figure}
  \centering
	\includegraphics[width = 0.75\textwidth]{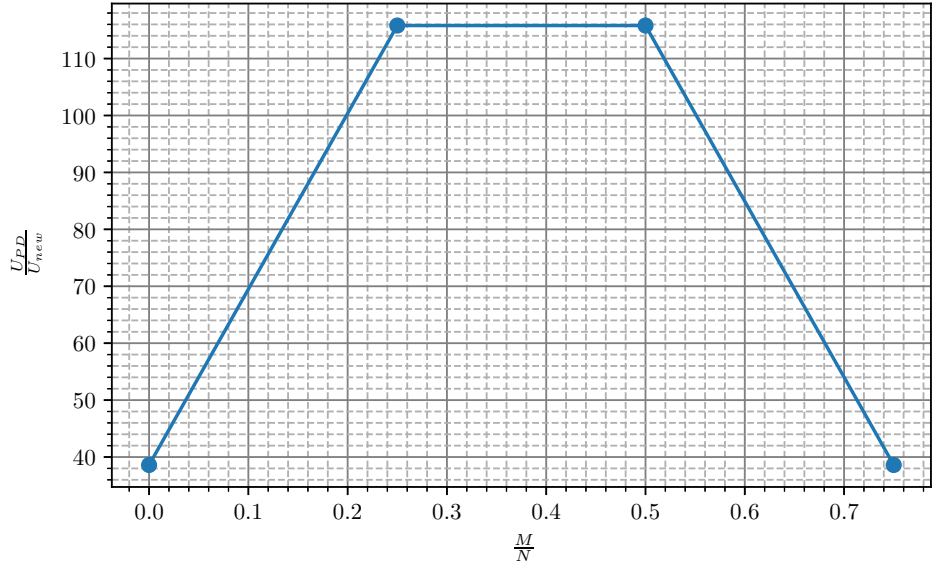}
  \caption{Ratio of upload costs for the Product design and MAN-PDA based MuPIR strategy with $B=2$ servers, $N=4$ files and $K = 4$ users for various values of $M/N$.}
  \label{fig:man-pda-vs-pd-upload}
\end{figure}

Although the subpacketization level in our scheme is still very high (increasing as $\binom{K}{\frac{KM}{N}}$), we can see that the subpacketization is only linear in $B$ in our case, whereas in the product design the subpacketization is proportional to $B^{N}$. Also, for a fixed caching ratio $M/N$, the subpacketization level in our scheme is independent of the number of files $N$, whereas, in product design, the upload cost increases exponentially with $N$. The $\binom{K}{KM/N}$ factor in our upload cost is due to the fact that we used MAN-PDA to construct our scheme, which has subpacketization level $\binom{K}{KM/N}$. Using PDAs with better subpacketization levels will improve the subpacketization in our case also.

\begin{table}[h!]
  \centering
  \begin{tabular}{| c | c | c |}
    \hline
    Parameter & MAN-PDA scheme & Product Design \\
    \hline
    \hline
    Rate & $\frac{K-\frac{KM}{N}}{\frac{KM}{N}+1}\RPIRg{B}{(\frac{KM}{N}+1)(N-1)}$ & $\frac{K-\frac{KM}{N}}{\frac{KM}{N}+1}\RPIR{B}{N}$ \\
    \hline
    Subpacketization & $(B-1)\binom{K}{\frac{KM}{N}}$ & $B^{N} \binom{K}{\frac{KM}{N}}$ \\
    \hline
    Upload cost & $BK(N-1)\log_{2}B$ bits & $BN \binom{K}{\frac{KM}{N} + 1} (\frac{KM}{N} + 1) \log_{2}\Big( \frac{B^{N}!}{B^{N-1}!} \Big)$ \\
    \hline
  \end{tabular}
  \caption{Comparing rate, subpacketization level and upload cost of MAN-PDA based scheme to the product design.}\label{tab:MN-PDAvsPD}
\end{table}

Now we show the the rate achieved in our scheme with the MAN-PDA is order optimal in terms of rate. Firstly, we define the optimal rate for a cache-aided MuPIR system with cache size $M$ to be $R^{*}(M)$. Recall that $R_{cc}(M)$, defined in~\eqref{eq:MAN_rate}, is the rate achieved by the MAN scheme for a coded caching setup, $R^{*}_{cc}(M)$ is the optimal rate for a coded caching problem, and $R_{ucc}^{*}(M)$ is the optimal rate for a coded caching problem under uncoded placement constraint defined in~\eqref{eq:CCoptimality} and~\eqref{eq:unCCoptimality} respectively.
\begin{theorem}\label{th:order_optimality}
  Given a PDA corresponding to the MAN scheme, the rate achieved in Theorem~\ref{th:achievability} is order optimal within a multiplicative factor of $8$ and order optimal within a multiplicative factor of $2$ considering $N \geq K$ and uncoded placement, i.e.,
  \begin{align*}
    R(M) &\leq 2 R^{*}_{ucc}(M) \\
    R(M) &\leq 8 R^{*}_{cc}(M).
  \end{align*}
\end{theorem}
\begin{proof}
  The rate achieved in Theorem~\ref{th:achievability} corresponding to the MAN-PDA is $R(M)$. The achievable scheme provided in Section~\ref{sec:sub:general-scheme} has uncoded placement. Therefore,
  \begin{align*}
    R(M) &= R_{cc}(M) \RPIRg{B}{(\frac{KM}{N}+1)(N-1)} \\
         &\leq R_{cc}(M) \frac{B}{B-1} \\
         &\stackrel{\eqref{eq:unCCoptimality}}{=} R_{ucc}^{*}(M) \frac{B}{B-1} \\
         &\stackrel{\eqref{eq:CCoptimality}}{\leq} 4R_{cc}^{*}(M) \frac{B}{B-1}.
  \end{align*}
  Now, as we have $B \geq 2 \implies \frac{B}{B-1} \leq 2$ and the $R_{cc}^{*}(M) \leq R^{*}(M)$ and $R_{ucc}^{*}(M) \leq R^{*}(M)$ because the optimal rate with PIR constraint cannot be less than the optimal rate without the PIR constraint, we have
  \begin{align*}
    R(M) &\leq 2 R^{*}_{ucc}(M) \mbox{ and } \\
    R(M) &\leq 8 R^{*}_{cc}(M).
  \end{align*}
\end{proof}

\subsection{Schemes based on PDAs described in~\cite{QYan17PDA}}\label{sec:sub:lowSubPDA}
As described in Section~\ref{sec:sub:CAPIR_intro}, the product design  proposed in~\cite{Ming21CaMuPIR} for general $N$, $K$, $B$ achieves rate
\begin{align*}
  R_{\mathrm{ PD}}(M) &= \frac {K-t}{t+1} \RPIR{B}{N}
\end{align*}
with subpacketization level $F_{PD} = B^{N}\binom{K}{t}$ and upload cost $U_{PD} = (t+1)\binom{K}{t+1}BN\log_{2}(\frac{B^{N}!}{B^{N-1}!})$.

For comparison, consider a multi-user PIR scheme based on the PDA constructions given in~\cite{QYan17PDA}, and described in Section~\ref{sec:sub:PDA_intro}. For $M/N = 1/q$ we have a $(m+1)\mhyphen\big( q(m+1),q^{m},q^{m-1},q^{m+1}-q^{m} \big)$ PDA and for $M/N = 1 - (1/q)$ we have a $(q-1)(m+1)\mhyphen\big(q(m+1),(q-1)q^{m}, (q-1)^{2}q^{m-1},q^{m}\big)$ PDA. According to Theorem~\ref{th:achievability}, $K = q(m+1)$ cache equipped users can privately retrieve their desired files from $B$ servers, with rate
\begin{align}\label{eq:qRate}
	R_{new}\Big(\frac{1}{q}\Big) &= (q - 1) \RPIR{B}{(m+1)(N-1)} \\
	R_{new}\Big(1 - \frac{1}{q}\Big) &= \frac{1}{(q - 1)} \RPIR{B}{(q-1)(m+1)(N-1)}.
\end{align}

In Table~\ref{tab:pda-vs-pd} we present the rate, subpacketization level and upload cost for the MuPIR schemes based on PDAs stated above and the product design.

\begin{table}[h!]
  \centering
  \begin{tabular}{| c | c | c | c |}
    \hline
    Parameter & Rate & Subpacketization & Upload cost \\
    \hline
    \hline
    PDA $\frac{M}{N} = \frac{1}{q}$ & $\frac{N-M}{M} \RPIRgshort{B}{\frac{KM}{N}(N-1)}$ & $(B-1)\frac{M}{N}{\Big( \frac{N}{M} \Big)}^{\frac{KM}{N}-1}$ & $BK(N-1)\log_{2}B$ \\
    \hline
    PDA $\frac{M}{N} = 1 - \frac{1}{q}$ & $\frac{N-M}{M} \RPIRgshort{B}{\frac{KM}{N}(N-1)}$ & $(B-1)\frac{M}{N}{\Big(\frac{N}{N-M}\Big)}^{K - \frac{KM}{N}}$ & $BK(N-1)\log_{2}B$ \\
    \hline
    Product Design & $\frac{K-\frac{KM}{N}}{\frac{KM}{N}+1}\RPIRshort{B}{N}$& $B^{N} \binom{K}{\frac{KM}{N}}$ & $\binom{K}{t + 1} (t + 1) \log_{2}{\Big( \frac{B^{N}!}{B^{N-1}!} \Big)}^{BN}$ \\
    \hline
  \end{tabular}
  \caption{Comparing rate, subpacketization level and upload cost of MuPIR schemes based on PDAs provided in~\cite{QYan17PDA} to the product design~\cite{Ming21CaMuPIR}.}\label{tab:pda-vs-pd}
\end{table}

We compare the rate achieved in our scheme using the above mentioned PDAs to the rate achieved using product design for a setup with $B = 10$ servers, $N = 18$ files and $q = m = 3$. In such a  setup we have $K = 12$ users. For $\frac{M}{N} = \frac{1}{3}$ we have a $(12 , 27 , 9 , 54)$ PDA and for $\frac{M}{N} = \frac{2}{3}$ we have a $(12 , 54 , 36 , 27)$ PDA (from~\cite{QYan17PDA}). For $M = 0$, when users don't have access to any cache, we consider the trivial $(12 , 1 , 0 , 12)$ PDA. When $\frac{M}{N} = 1$, users can store all the files in their cache, and servers need not transmit anything. For this example, rates for the product design and PDA-based schemes are plotted in Fig.~\ref{fig:rate-compare1}.

\begin{figure}
  \centering
	\includegraphics[width = 0.7\textwidth]{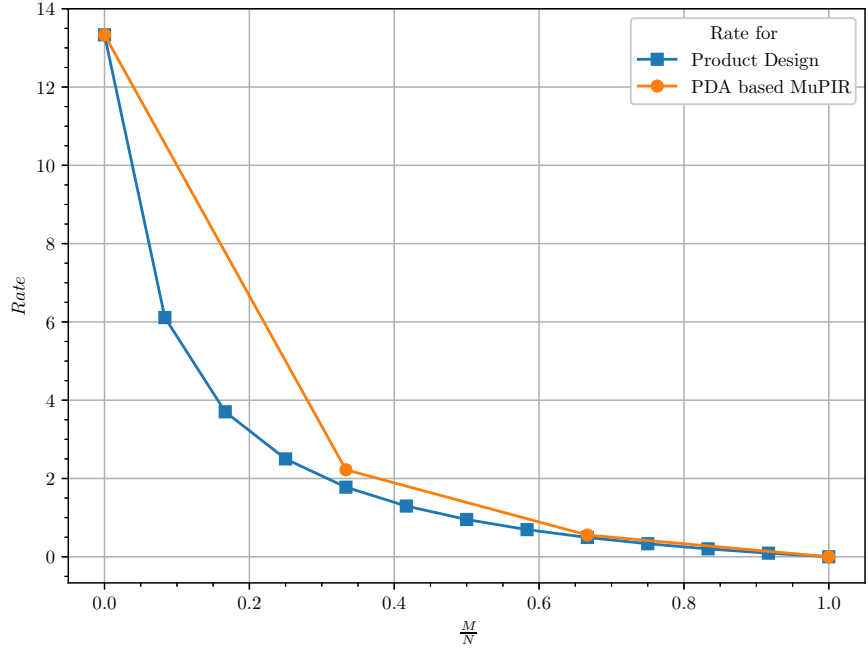}
  \caption{Comparing rates for Product design and PDA based MuPIR strategy with $B=10$ servers, $N=18$ files and $K = 12$ users for various values of $M/N$.}
  \label{fig:rate-compare1}
\end{figure}
The rate achieved using the PDA based scheme is higher that the rate achieved by product design. But, as we will see, the PDA based scheme offer better performance in terms of upload cost and subpacketization level, whereas the rate is only marginally worse in our case.

To compare these two systems, we will consider PDAs with $q = 3$ and we will vary $m$. This is equivalent to increasing the number of users $K = q(m+1)$ in the problem. The caching ratio $M/N = 1/q$ remains constant, but the total system memory i.e. $t = KM/N = m+1$ increases with $m$. Again we will consider $B = 10$ servers and $N = K = q(m+1)$ files.

The rate is given in \eqref{eq:qRate}, the subpacketization level is $F_{new} = q^{m}$ and the upload cost is $U_{new} = BK(N-1)\log_{2}B$ bits. Also note that $t = \frac{KM}{N} = \frac{q(m+1)}{q} = m + 1$.

Consider the quantity
\begin{align*}
	\frac{R_{new}}{R_{PD}} &= \frac{ \frac{S}{F} \Bigg( 1 + \frac{1}{S}\sum_{s \in [S]} \RPIRs{B}{|\Kappa_{s}| (N-1)} \Bigg) }{\frac {K-t}{t+1} \RPIR{B}{N}} \\
  &= \frac{ (q-1) \RPIR{B}{(N-1)(m+1)} }{\frac{q(m+1)-(m+1)}{m+2} \RPIR{B}{N}} \\
  &= \frac{m+2}{m+1} \times \frac{1 - B^{m - mN - N}}{1 - B^{-N}} \\
  &= \frac{KM/N + 1}{KM/N} \times \frac{1 - B^{-\frac{KM}{N}(N-1) - 1}}{1 - B^{-N}}.
\end{align*}
As $K \rightarrow \infty$ we can see that $\frac{R_{new}}{R_{PD}} \rightarrow \frac{1}{1 - B^{-N}}$. Again, if number of servers and/or number of files stored in the servers are large, the ratio $\frac{1}{1 - B^{-N}} \rightarrow 1$. For $q=3$, $B=10$ servers, $N=300$ files and $\frac{M}{N} = \frac{1}{3}$ and varying number of users we compare the rates in Fig.~\ref{fig:rate-compareK}
\begin{figure}
  \centering
	\includegraphics[width = 0.7\textwidth]{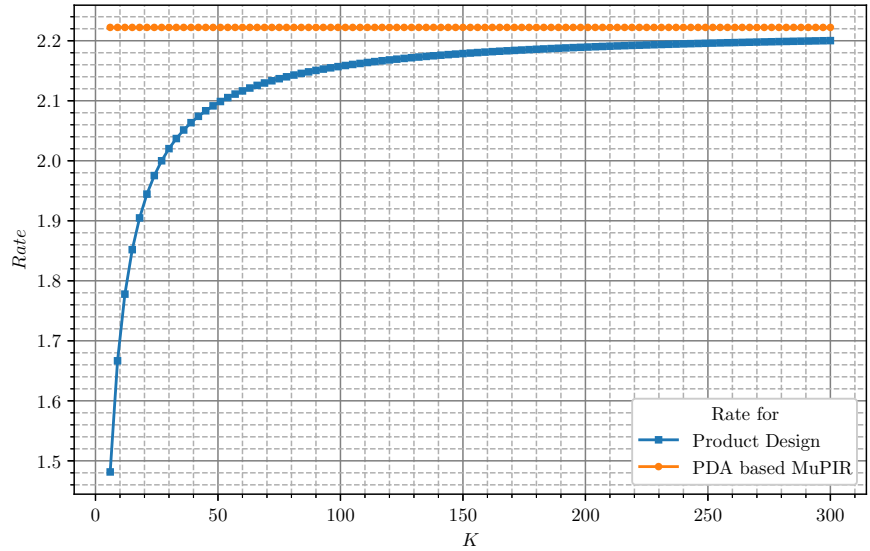}
  \caption{Comparing rates for Product design and PDA based MuPIR strategy for $B=10$ servers, $M/N = 1/3$, $N = 300$ files and varying number of users.}
  \label{fig:rate-compareK}
\end{figure}

For comparing subpacketization in these two schemes, consider the term
\begin{align*}
	\frac{F_{new}}{F_{PD}} &= \frac{(B-1) q^{m}}{B^{N}\binom{K}{t}} = \frac{(B-1){\frac{N}{M}}^{\frac{KM}{N} - 1}}{B^{N}\binom{K}{KM/N}}.
\end{align*}
For given $M$, $N$ and $K$ we can see that $\frac{F_{new}}{F_{PD}} \propto \frac{B-1}{B^{N}}$, and therefore $\frac{F_{new}}{F_{PD}} \rightarrow 0$ as the number of servers, $B$, increases. Also, for a given caching ratio, $M/N$, $\frac{{(N/M)}^{(KM/N) - 1}}{\binom{K}{KM/N}} \rightarrow 0$ as $K \rightarrow \infty$.

Again consider the example with $B = 10$ servers, $N = 18$ files and $q = 3$. Here we will vary the number of users by changing $m$. We will consider $M/N = 1/q = 1/3$. In Fig.~\ref{fig:subpacket-compare1} we plot $F_{new}$ and $F_{PD}$ against $K$. We can see that the subpacketization level is very high for the product design proposed in~\cite{Ming21CaMuPIR} compared to the PDA-based strategy of Theorem~\ref{th:achievability}.
\begin{figure}
  \centering
	\includegraphics[width = 0.7\textwidth]{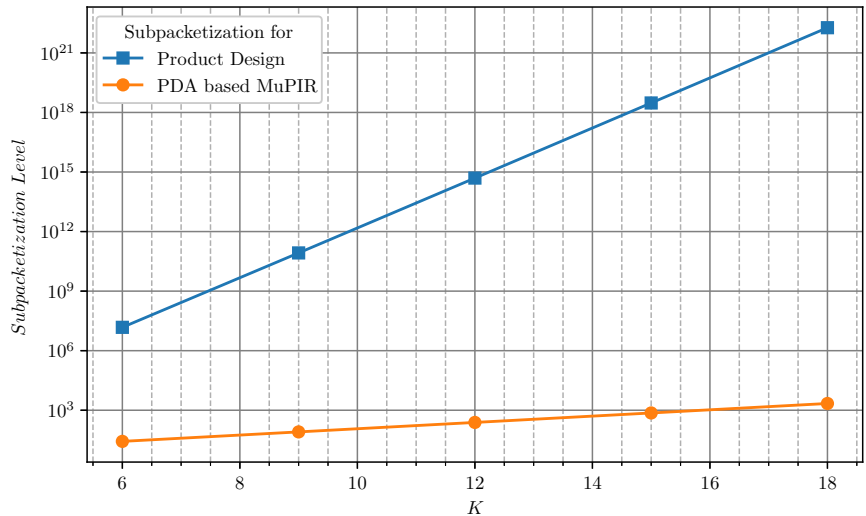}
  \caption{Comparing subpacketization level for Product design and PDA based MuPIR strategy for $B=10$ servers, $M/N = 1/3$ and varying number of users considering $N=K$.}
  \label{fig:subpacket-compare1}
\end{figure}

To compare the upload costs, consider the term
\begin{align*}
	\frac{U_{new}}{U_{PD}} &= \frac{ BK(N-1)\log_{2}B }{ (\frac{KM}{N}+1)\binom{K}{\frac{KM}{N}+1}BN\log_{2}(\frac{B^{N}!}{B^{N-1}!}) } \\
  &= \frac{K}{(\frac{KM}{N} + 1) \binom{K}{\frac{KM}{N} + 1}} \times \frac{N-1}{N} \times \frac{\log(B)}{ \log(\frac{B^{N}!}{B^{N-1}!})}.
\end{align*}
Here, for some fixed caching ratio, $M/N$, number of servers $B$ and number of files $N$, we have, $\frac{U_{new}}{U_{PD}} \propto \frac{K}{(\frac{KM}{N} + 1) \binom{K}{\frac{KM}{N} + 1}}$, therefore, $\frac{U_{new}}{U_{PD}} \rightarrow 0$ as $K$ increases. Also, for a fixed caching ratio $M/N$ and number of users $K$ consider
\begin{align*}
	\frac{U_{new}}{U_{PD}} &\propto \frac{(N-1) \log(B)}{ N \log(\frac{B^{N}!}{B^{N-1}!}) } \\
                         &= \frac{(N-1) \log(B)}{N (\log(B^{N}) + \cdots + \log(B^{N-1} + 1))} \\
                         &= \frac{(N-1) \log(B)}{N \sum_{n = 1}^{B^{N} - B^{N-1}}(\log(B^{N-1} + n))} \\
  &< \frac{(N-1) \log(B)}{N \sum_{n = 1}^{B^{N} - B^{N-1}}(\log(B^{N-1}))} = \frac{(N-1) \log(B)}{N (B^{N} - B^{N-1}) (\log(B^{N-1}))} \\
  &= \frac{1}{N(B^{N} - B^{N-1})}.
\end{align*}
So, if either the number of files or the number of servers increases, we have $\frac{U_{new}}{U_{PD}} \rightarrow 0$.

Consider a cache-aided multi-user setup with $B = 2$ servers, $N = 18$ files and $q = 3$. Here we will vary the number of users by changing $m$. We will consider $M/N = 1/q = 1/3$. In Fig.~\ref{fig:upload-compare1} we plot $U_{new}$ and $U_{PD}$ against $K$. We can see that the upload cost is very high for the product design proposed in~\cite{Ming21CaMuPIR} compared to the PDA-based strategy of Theorem~\ref{th:achievability}.
\begin{figure}
  \centering
	\includegraphics[width = 0.7\textwidth]{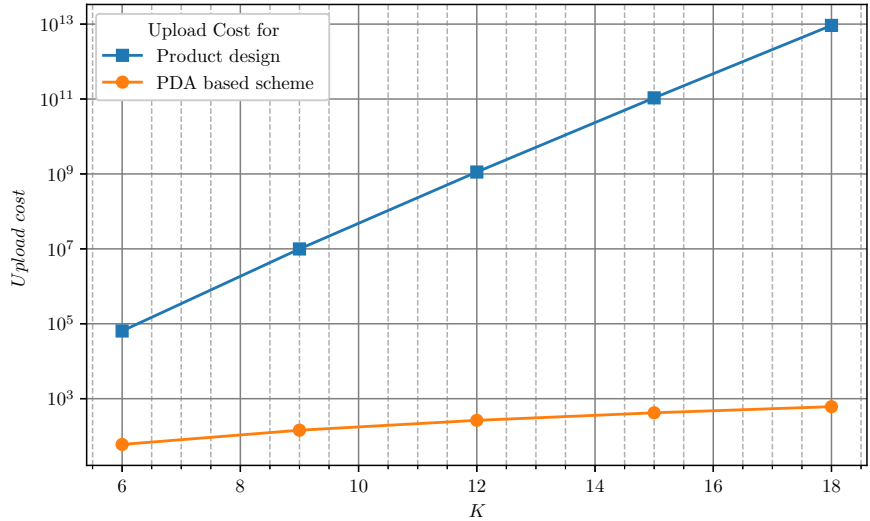}
  \caption{Comparing upload cost for Product design and PDA based MuPIR strategy for $B=10$ servers, $M/N = 1/3$ and varying number of users considering $N=K$.}
  \label{fig:upload-compare1}
\end{figure}

\subsection{Schemes based on other known PDAs}\label{sec:sub:other-PDA}
In this section, we summarize some other known PDAs and briefly discuss MuPIR results based on those PDAs.

In~\cite{Yan18EdgeColorPDA} for any $m,a,b,\lambda \in \Nbb$ with $0<a,b<m,0 \leq \lambda \leq \min\{a,b\}$ , a $g-(K,F,Z,S)$ PDA construction is provided, where
\begin{align*}
  K=&{\binom{m}{ a}},\quad F={\binom{m}{ b}},~Z={\binom{m}{ b}}-{\binom{a}{ \lambda }}{\binom{m-a}{ b-\lambda }}, \\
  S=&{\binom{m}{ a+b-2\lambda }}\cdot \min \left \{{{\binom{m\!-\!a\!-\!b\!+\!2\lambda }{ \lambda }},{\binom{a\!+\!b\!-\!2\lambda }{ a-\lambda }}}\right \}\!,
\end{align*}
and
\begin{equation*}
  g=\max \left \{{{\binom{a+b-2\lambda }{ a-\lambda }},{\binom{m-a-b+2\lambda }{ \lambda }}}\right \}\!.
\end{equation*}

In~\cite{Li22CombinatorialPDA}, the following three PDA constructions are provided using combinatorial designs.

{\it Construction 1~\cite[Theorem 4]{Li22CombinatorialPDA}:} If there exist a simple $(v,k,2)$-SBIBD then there exist a $k-(vk,v,v-k+1,v(k-1))$ PDA.

{\it Construction 2~\cite[Theorem 5]{Li22CombinatorialPDA}:} If there exist a $t-(v,k,1)$ design for any $i\in[1,t-1]$, then there exist a $\frac {\binom {v-i}{t-i}}{\binom {k-i}{t-i}} - \left({\frac {\binom v{t}\binom {k}{i}}{\binom {k}{t}},\binom {v}{t-i}, \binom {v}{t-i}-\binom {k-i}{t-i},\,\,\binom {k-i}{t-i}\binom {v}{i}}\right)$ PDA.

{\it Construction 3~\cite[Theorem 6]{Li22CombinatorialPDA}:} If there exist a simple $t-(v,k,\lambda)$ design and $k \leq 2t$ then there exist a $\frac {\lambda \binom {v}{t}}{\binom {v}{k-t}} - \left({\binom v{t},\frac {\lambda \binom v{t}}{\binom k{t}},\frac {\lambda \binom v{t}}{\binom k{t}}-\lambda,\binom v{k-t}}\right)$ PDA.

In Table~\ref{tab:known-pdas}, we enumerate the results of MuPIR schemes based on the PDAs mentioned above using the definition
\[
  R_{PIR}(B , N) \triangleq \left( 1 + \frac{1}{B} + \cdots + \frac{1}{B^{N}}\right).
\]

\begin{table}[h!]
  \centering
  \begin{tabular}{| p{0.3\textwidth} | p{0.3\textwidth} | c | c |}
    \hline
    PDA Description & Rate & Subpacketization & Upload cost \\
    \hline
    \hline
    {\tiny $\Big({\binom{m}{ a}},{\binom{m}{ b}},{\binom{m}{b}}-{\binom{a}{ \lambda }}{\binom{m-a}{ b-\lambda }} , {\binom{m}{ a+b-2\lambda }}\cdot \min \left \{{{\binom{m\!-\!a\!-\!b\!+\!2\lambda }{ \lambda }},{\binom{a\!+\!b\!-\!2\lambda }{ a-\lambda }}}\right \}\!\Big)$}
    PDA where $m,a,b,\lambda \in \Nbb$ with $0<a, b<m, 0\leq \lambda \leq \min \{a,b\}$~\cite{Yan18EdgeColorPDA} & {\tiny $\frac{{\binom{m}{ a+b-2\lambda }}\cdot \min \{{{\binom{m\!-\!a\!-\!b\!+\!2\lambda }{ \lambda }},{\binom{a\!+\!b\!-\!2\lambda }{ a-\lambda }}}\}}{\binom{m}{b}} \cdot R_{PIR}\Big(B , (N-1)\cdot \newline \max \left \{{{\binom{a+b-2\lambda }{ a-\lambda }}, {\binom{m-a-b+2\lambda }{ \lambda }}} \right \}\Big)$} & $(B-1)\binom{m}{b}$ & $B\binom{m}{a}(N-1)\log_{2}B$ \\
    \hline
    $(vk,v,v-k+1,v(k-1))$ PDA from a $v,k,2$-SBIBD~\cite[Theorem 4]{Li22CombinatorialPDA} & $(k-1)R_{PIR}(B , k(N-1))$ & $v(B-1)$ & $Bvk(N-1)\log_{2}B$ \\
    \hline
    {\tiny $\left({\frac {\binom{v}{t} \binom{k}{i}} {\binom {k}{t}},\binom {v}{t-i}, \binom {v}{t-i}-\binom {k-i}{t-i}, \binom {k-i}{t-i}\binom {v}{i}}\right)$} PDA from a $t-(v,k,1)$-design~\cite[Theorem 5]{Li22CombinatorialPDA} & $\frac{\binom {k-i}{t-i}\binom {v}{i}}{\binom{v}{t-i}} \cdot R_{PIR}\left(B , (N-1)\frac {\binom {v-i}{t-i}}{\binom {k-i}{t-i}}\right)$ & $\binom{v}{t-i}(B-1)$ & $B\frac {\binom{v}{t} \binom{k}{i}} {\binom {k}{t}}(N-1)\log_{2}B$ \\
    \hline
    $\left({\binom v{t},\frac {\lambda \binom v{t}}{\binom k{t}},\frac {\lambda \binom v{t}}{\binom k{t}}-\lambda,\binom v{k-t}}\right)$ PDA from $t-(v , k , \lambda)$ design where $k \leq 2t$~\cite[Theorem 6]{Li22CombinatorialPDA}. & $\frac{\binom{v}{k-t}}{\lambda \binom{v}{t} / \binom{k}{t}} \cdot R_{PIR}\left(B , \frac {\lambda \binom {v}{t}}{\binom {v}{k-t}}(N-1)\right)$ & $(B-1)\frac {\lambda \binom v{t}}{\binom k{t}}$ & $B\binom{v}{t}(N-1)\log_{2}B$ \\
    \hline
  \end{tabular}
  \caption{Some known PDAs and MuPIR results on schemes based on those PDAs.}\label{tab:known-pdas}
\end{table}

\section{Our Achievable Scheme}\label{sec:scheme-section}

In this section, we describe an achievable scheme for cache-aided multi-user PIR setup, using a $(K,F,Z,S)$ PDA, thus proving Theorem~\ref{th:achievability}. We first describe the achievable scheme using an example.

\subsection{Example}\label{sec:sub:example-scheme}
Consider a cache-aided MuPIR system with $B = 3$ servers, storing $N=6$ files $\{W_{0} , \ldots , W_{5}\}$. We consider a cache-aided setup corresponding to the following PDA:
\begin{equation}
  \Pbf =\left [{\begin{array}{cccccc}
                  *& *& 1& *& 2& 3\\
                  *& 1& *& 2& *& 4\\
                  1& *& *& 3& 4& *\\
                  2& 3& 4& *& *& *
                \end{array} }\right ].
\end{equation}
There are $K = 6$ cache-equipped users, each capable of storing three files in their cache.
{\bf Placement Phase}: Each file is divided into four non-overlapping equal subfiles
\begin{align*}
 W_{0} &= \{W_{0}^{1} , W_{0}^{2} , W_{0}^{3} , W_{0}^{4}\} , & W_{1} &= \{W_{1}^{1} , W_{1}^{2} , W_{1}^{3} , W_{1}^{4}\}, \\
 W_{2} &= \{W_{2}^{1} , W_{2}^{2} , W_{2}^{3} , W_{2}^{4}\} , & W_{3} &= \{W_{3}^{1} , W_{3}^{2} , W_{3}^{3} , W_{3}^{4}\}, \\
 W_{4} &= \{W_{4}^{1} , W_{4}^{2} , W_{4}^{3} , W_{4}^{4}\} , & W_{5} &= \{W_{5}^{1} , W_{5}^{2} , W_{5}^{3} , W_{5}^{4}\}.
\end{align*}
Then, subfile $f$ of every file is placed in the cache of user $k$ if $\Pbf_{f , k} = *$. The content stored in the cache of the users is as follows:
\begin{align*}
  &\Zcal_{1} = \{W_{n}^{1} , W_{n}^{2} : n \in [0:5]\} , &  \Zcal_{2} &= \{W_{n}^{1} , W_{n}^{3} : n \in [0:5]\}, \\
  &\Zcal_{3} = \{W_{n}^{2} , W_{n}^{3} : n \in [0:5]\} , &  \Zcal_{4} &= \{W_{n}^{1} , W_{n}^{4} : n \in [0:5]\}, \\
  &\Zcal_{5} = \{W_{n}^{2} , W_{n}^{4} : n \in [0:5]\} , &  \Zcal_{5} &= \{W_{n}^{3} , W_{n}^{4} : n \in [0:5]\}.
\end{align*}
{\bf Private Delivery Phase}: In this phase, every user independently and uniformly chooses a file to retrieve from the servers. Let the demand of user $k$ be denoted by $d_{k}$. We consider
\begin{align*}
	d_{1} &= 3 , &  d_{2} &= 1 , &  d_{3} &= 0 , \\
	d_{4} &= 4 , &  d_{5} &= 5 , &  d_{1} &= 1
\end{align*}
i.e.\ the demand vector is $\dbf = (3 , 1 , 0 , 4 , 5 , 1)$. Now every user generates a uniformly random vector of length $N-1 = 5$ over $[0:B-1] = [0:2]$. User $k$ generates $\Vbf^{k} = (V_{0}^{k} \ldots V_{4}^{k})$, let
\begin{align*}
	\Vbf^{1} &= (1 , 0 , 1 , 2 , 0),   &   \Vbf^{2} &= (0 , 1 , 1 , 0 , 1), \\
	\Vbf^{3} &= (1 , 2 , 2 , 0 , 2),   &   \Vbf^{4} &= (0 , 0 , 1 , 2 , 2), \\
	\Vbf^{5} &= (0 , 0 , 1 , 0 , 2),   &   \Vbf^{6} &= (0 , 1 , 0 , 1 , 0).
\end{align*}
We also note the values of $\overline{V}^{k} \triangleq {\Big( -\sum_{n \in [0:4]} V_{n}^{k} \Big)}_{B}, \forall k \in [6]$.
\begin{align*}
	\overline{V}^{1} &= 2,   &   \overline{V}^{2} &= 0, \\
	\overline{V}^{3} &= 2,   &   \overline{V}^{4} &= 1, \\
	\overline{V}^{5} &= 0,   &   \overline{V}^{6} &= 1.
\end{align*}
Now, each user constructs three $N = 6$ length vectors for every server over $[0:B-1] = [0:2]$. Specifically, consider user $1$ wanting file $W_{3}$. User $1$ generates three vectors as follows:
\begin{align*}
	\Qbf_{0}^{1} &= (V^{0} , V^{1} , V^{2} , {\Big(0 + \overline{V}^{1}\Big)}_{B} , V^{3} , V^{4}) = (1 , 0 , 1 , {\bf 2} , 2 , 0),\\
	\Qbf_{1}^{1} &= (V^{0} , V^{1} , V^{2} , {\Big(1 + \overline{V}^{1}\Big)}_{B} , V^{3} , V^{4}) = (1 , 0 , 1 , {\bf 0} , 2 , 0),\\
	\Qbf_{2}^{1} &= (V^{0} , V^{1} , V^{2} , {\Big(2 + \overline{V}^{1}\Big)}_{B} , V^{3} , V^{4}) = (1 , 0 , 1 , {\bf 1} , 2 , 0).
\end{align*}
Note that only the $4^{th}$ coordinate (which corresponds to file $W_{3}$ considering 0-indexing) of these three vectors are different, also taking the sum of each of the elements of $\Qbf_{b}^{1}$ modulo $B$ is equal to $b$ for any $b \in [0:2]$. User $1$ sends $\Qbf_{b}^{1}$ to server $b, \forall b \in [0:B-1]$. Table~\ref{tab:query_table} contains all the queries sent to the servers by the users.

\begin{table}
  \centering
  \begin{tabular}{| c | c | c | c |}
    \hline
    $\Qbf_{b}^{k}$ &  Server $0$                &  Server $1$                     &  Server $2$ \\
    \hline
    User $1$  &  (1 , 0 , 1 , {\bf 2} , 2 , 0)  &  (1 , 0 , 1 , {\bf 0} , 2 , 0)  &  (1 , 0 , 1 , {\bf 1} , 2 , 0) \\
    User $2$  &  (0 , {\bf 0} , 1 , 1 , 0 , 1)  &  (0 , {\bf 1} , 1 , 1 , 0 , 1)  &  (0 , {\bf 2} , 1 , 1 , 0 , 1) \\
    User $3$  &  ({\bf 2} , 1 , 2 , 2 , 0 , 2)  &  ({\bf 0} , 1 , 2 , 2 , 0 , 2)  &  ({\bf 1} , 1 , 2 , 2 , 0 , 2) \\
    User $4$  &  (0 , 0 , 1 , 2 , {\bf 1} , 2)  &  (0 , 0 , 1 , 2 , {\bf 2} , 2)  &  (0 , 0 , 1 , 2 , {\bf 0} , 2) \\
    User $5$  &  (0 , 0 , 1 , 0 , 2 , {\bf 0})  &  (0 , 0 , 1 , 0 , 2 , {\bf 1})  &  (0 , 0 , 1 , 0 , 2 , {\bf 2}) \\
    User $6$  &  (0 , {\bf 1} , 1 , 0 , 1 , 0)  &  (0 , {\bf 2} , 1 , 0 , 1 , 0)  &  (0 , {\bf 0} , 1 , 0 , 1 , 0) \\
    \hline
  \end{tabular}
  \caption{Queries sent by the users to the servers.}\label{tab:query_table}
\end{table}

After receiving their respective queries, the servers construct answers based on these queries, the files and the PDA. Consider again the PDA $\Pbf$, every server construct a transmission for every $s \in [S] = [4]$. Servers divide each subfile into $B-1 = 2$ sub-subfiles.
\[
  W_{n}^{f} = (W_{n , 1}^{f} , W_{n , 2}^{f}) ,~\forall (n , f) \in [0:5] \times [4].
\]
We also define $W_{n,0}^{f} = 0$ for any $n \in [0:5]$ and any $f \in [4]$. Consider $s = 1$, note that $\Pbf_{f,k} = 1$ only for $(f,k) \in \{(1,3) , (2,2) , (3,1)\}$. Server $b$ computes
\[
  \sum_{(f,k) \in \{(1,3) , (2,2) , (3,1)\}} \sum_{n = 0}^{5} W_{n , \Qbf_{b , n}^{k}}^{f}
\]
and broadcasts it to the users. We enumerate the broadcasts corresponding to $s = 1$ in Table~\ref{tab:answer_table}.
\begin{table}
	\centering
  \begin{tabular}{| p{0.2\textwidth} | p{0.2\textwidth} | p{0.2\textwidth} |}
    \hline
    Server $0$ & Server $1$ & Server $2$ \\
    \hline
    $({\bf W_{0,2}^{1}} + W_{1,1}^{1} + W_{2,2}^{1} \newline + W_{3,2}^{1} + W_{4,0}^{1} + W_{5,2}^{1}) \newline + (W_{0,0}^{2} + {\bf W_{1,0}^{2}} + W_{2,1}^{2} \newline + W_{3,1}^{2} + W_{4,0}^{2} + W_{5,1}^{2}) \newline + (W_{0,1}^{3} + W_{1,0}^{3} + W_{2,1}^{3} \newline + {\bf W_{3,2}^{3}} + W_{4,2}^{3} + W_{5,0}^{3})$  &%
    $({\bf W_{0,0}^{1}} + W_{1,1}^{1} + W_{2,2}^{1} \newline + W_{3,2}^{1} + W_{4,0}^{1} + W_{5,2}^{1}) \newline + (W_{0,0}^{2} + {\bf W_{1,1}^{2}} + W_{2,1}^{2} \newline + W_{3,1}^{2} + W_{4,0}^{2} + W_{5,1}^{2}) \newline + (W_{0,1}^{3} + W_{1,0}^{3} + W_{2,1}^{3} \newline + {\bf W_{3,0}^{3}} + W_{4,2}^{3} + W_{5,0}^{3})$  &%
    $({\bf W_{0,1}^{1}} + W_{1,1}^{1} + W_{2,2}^{1} \newline + W_{3,2}^{1} + W_{4,0}^{1} + W_{5,2}^{1}) \newline + (W_{0,0}^{2} + {\bf W_{1,2}^{2}} + W_{2,1}^{2} \newline + W_{3,1}^{2} + W_{4,0}^{2} + W_{5,1}^{2}) \newline + (W_{0,1}^{3} + W_{1,0}^{3} + W_{2,1}^{3} \newline + {\bf W_{3,1}^{3}} + W_{4,2}^{3} + W_{5,0}^{3})$ \\
    \hline
  \end{tabular}
  \caption{Answers broadcasted by the servers corresponding to $s = 1$}\label{tab:answer_table}
\end{table}

Now we will demonstrate that users $1$, $2$ and $3$ will retrieve subfiles $3$, $2$ and $1$ of their desired files, respectively. As user $3$ has access to subfiles $2$ and $3$ of every file, it can cancel out the interference terms from these transmissions and be left with the terms $A_{b,s}^{f,k}, b \in [0:B-1]$ defined as:
\begin{align*}
  &{\bf W_{0,2}^{1}} + W_{1,1}^{1} + W_{2,2}^{1} + W_{3,2}^{1} + W_{4,0}^{1} + W_{5,2}^{1} \triangleq A_{0,1}^{3,1}, \\
  &{\bf W_{0,0}^{1}} + W_{1,1}^{1} + W_{2,2}^{1} + W_{3,2}^{1} + W_{4,0}^{1} + W_{5,2}^{1} \triangleq A_{1,1}^{3,1} \mbox{ and} \\
  &{\bf W_{0,1}^{1}} + W_{1,1}^{1} + W_{2,2}^{1} + W_{3,2}^{1} + W_{4,0}^{1} + W_{5,2}^{1} \triangleq A_{2,1}^{3,1}.
\end{align*}
Recalling that $W_{n,0}^{f} = 0$, the second term above is $A_{1,1}^{3,1} = W_{0,1}^{1} + W_{0,2}^{1} + W_{0,2}^{1} + W_{0,0}^{1} + W_{0,2}^{1}$. Therefore using this second term user 3 gets
\begin{align*}
	{\bf W_{0,2}^{1}} &= A_{0,1}^{1,3} - A_{1,1}^{3,1} \mbox{ and } \\
  {\bf W_{0,1}^{1}} &= A_{2,1}^{1,3} - A_{1,1}^{3,1}.
\end{align*}
 Similarly, user $2$ can cancel out the interference of subfiles $1$ and subfile $3$ and will remain with
\begin{align*}
  &W_{0,0}^{2} + {\bf W_{1,0}^{2}} + W_{2,1}^{2} + W_{3,1}^{2} + W_{4,0}^{2} + W_{5,1}^{2} \triangleq A_{0,1}^{2,2}, \\
  &W_{0,0}^{2} + {\bf W_{1,1}^{2}} + W_{2,1}^{2} + W_{3,1}^{2} + W_{4,0}^{2} + W_{5,1}^{2} \triangleq A_{1,1}^{2,2} \mbox{ and} \\
  &W_{0,0}^{2} + {\bf W_{1,2}^{2}} + W_{2,1}^{2} + W_{3,1}^{2} + W_{4,0}^{2} + W_{5,1}^{2} \triangleq A_{2,1}^{2,2}.
\end{align*}
Again, as ${\bf W_{1,0}^{2}} = 0$, user $2$ decodes its required subfiles as
\begin{align*}
	{\bf W_{1,1}^{2}} &= A_{1,1}^{2,2} - A_{0,1}^{2,2} \mbox{ and } \\
  {\bf W_{1,2}^{2}} &= A_{2,1}^{2,2} - A_{0,1}^{2,2}.
\end{align*}
User $1$ removes the interference of subfiles $1$ and subfile $2$ and is left with
\begin{align*}
  &W_{0,1}^{3} + W_{1,0}^{3} + W_{2,1}^{3} + {\bf W_{3,2}^{3}} + W_{4,2}^{3} + W_{5,0}^{3} \triangleq A_{0,1}^{1,3}, \\
  &W_{0,1}^{3} + W_{1,0}^{3} + W_{2,1}^{3} + {\bf W_{3,0}^{3}} + W_{4,2}^{3} + W_{5,0}^{3} \triangleq A_{1,1}^{1,3} \mbox{ and} \\
  &W_{0,1}^{3} + W_{1,0}^{3} + W_{2,1}^{3} + {\bf W_{3,1}^{3}} + W_{4,2}^{3} + W_{5,0}^{3} \triangleq A_{2,1}^{1,3},
\end{align*}
and decodes the desired subfiles as
\begin{align*}
	{\bf W_{3,1}^{3}} &= A_{2,1}^{1,3} - A_{1,1}^{1,3} \mbox{ and } \\
  {\bf W_{3,2}^{3}} &= A_{0,1}^{1,3} - A_{1,1}^{1,3}.
\end{align*}

Similarly, user $k$ for any $k \in [6]$ can decode its desired file. Specifically, user $k$ decode subfile $f$ from the transmission corresponding to that $s \in [4]$ for which $\Pbf_{f,k} = s$. Whereas, if $\Pbf_{f,k} = *$ then subfile $f$ is stored in the cache of user $k$.

Here, we see that every server is broadcasting a transmission of size equal to one sub-subfile for every $s \in [S] = [4]$. As every file is divided into $F = 4$ subfiles and every subfile is divided into $B-1 = 2$ sub-subfiles, every server is transmitting $L / 8$ bits for every $s \in [4]$. Therefore, the total number of bits transmitted is
\[
  \sum_{b \in [0:2]}R_{b}L = 1.5L \mbox{ bits}.
\]

\subsubsection*{Special cases}
Consider a special case when $\Vbf^{k} = (0,0,0,0,0) , \forall k \in \{1 , 2 , 3\}$, which can happen with a probability of $3^{-15}$. In that case, the queries sent to server $0$ is $\Qbf_{0}^{k} = (0,0,0,0,0,0), \forall k \in \{1 , 2 , 3\}$, and thus, the answer broadcasted by server $0$, for $s = 1$ will be
\begin{align*}
  &({\bf W_{0,0}^{1}} + W_{1,0}^{1} + W_{2,0}^{1} + W_{3,0}^{1} + W_{4,0}^{1} + W_{5,0}^{1}) \\
  + &(W_{0,0}^{2} + {\bf W_{1,0}^{2}} + W_{2,0}^{2} + W_{3,0}^{2} + W_{4,0}^{2} + W_{5,0}^{2}) \\
  + &(W_{0,0}^{3} + W_{1,0}^{3} + W_{2,0}^{3} + {\bf W_{3,0}^{3}} + W_{4,0}^{3} + W_{5,0}^{3}) = 0.
\end{align*}
This happens because $W_{n,0}^{f} = 0$ for any $n$ and $f$. As the answer computed by server $0$ is already determined by the queries, and is independent of the files, if $\Vbf^{k} = (0,0,0,0,0) , \forall k \in \{1 , 2 , 3\}$ then server $0$ will not broadcast anything. Similarly, if $\Vbf^{k} = (0,0,0,0,0) , \forall k \in \{1 , 4 , 5\}$, then, answers computed by server $0$ for $s = 2$ is predetermined by the queries itself, therefore server $0$ will not transmit anything for $s = 2$. Similarly, if $\Vbf^{k} = (0,0,0,0,0) , \forall k \in \{2 , 4 , 6\}$ or $\Vbf^{k} = (0,0,0,0,0) , \forall k \in \{3 , 5 , 6\}$, then server $0$ will not transmit for $s = 3$ or $s = 4$ respectively.

The average number of bits transmitted by server $0$ over all possible queries is
\begin{align*}
  \Ebb \Big\{ R_{0}L \Big\} &= 4 \times \Bigg( (3^{-15}) \cdot 0 + (1 - 3^{-15}) \cdot \frac{L}{8} \Bigg) \\
  \implies \Ebb \Big\{ R_{0}L \Big\} &= \frac{1 - 3^{-15}}{2}L.
\end{align*}
Transmissions from all other servers are dependent on the content of the files for all possible queries, so every server transmits $L/2$ bits irrespective of the realization of the queries.

\subsubsection*{Rate in this example} The average number of bits broadcasted by the servers is
\begin{align*}
	\Ebb \sum_{b = 0}^{2} R_{b}L &= \sum_{b=0}^{2} \Ebb \Big\{ R_{b}L \Big\} \\
                               &= \frac{1 - 3^{-15}}{2}L + L \\
                    \implies R &= \frac{3 - 3^{-15}}{2}
\end{align*}
which is equal to the rate provided in Theorem~\ref{th:achievability}.

\subsubsection*{Upload cost and Subpacketization Level} Every file is divided into $F = 4$ subfiles, and every subfile is divided into $B-1 = 2$ sub-subfiles. Therefore the subpacketization level is $8$. Every user is constructing one query for every server. Therefore, there are $BK = 18$ queries in total. Each query consists of $6$ integers, out of which $5$ integers are chosen uniformly from $[0:2]$, and the $6^{th}$ integer is dependent on the first five (the modulo $3$ sum of all six integers should be equal to the server index). Therefore, the upload cost is $90\log_{2}3$ bits.

\subsection{General Description}\label{sec:sub:general-scheme}
Consider a $(K , F , Z , S)$ PDA, $\Pbf$, and a cache-aided multi-user system with $B$ non-colluding servers each storing $N$ files. There are $K$ users, each equipped with a dedicated cache of size $ML$ bits where $M / N = Z / F$.
\subsubsection*{\textbf{Placement Phase}} During placement phase, each file is split into $F$ non-overlapping and equal subfiles i.e.
\[
  W_n = \{W_{n}^{f} : f \in [F]\} , \forall n \in [N].
\]
Then, the cache of user $k \in [K]$ is filled with
\[
  \Zcal_k = \{ W_{n}^{f} : \Pbf_{f , k} = *, \forall n \in [N]\}.
\]
Specifically, for user $k$ consider the $k$-th column of $\Pbf$. In the $k$-th column, if row $f$ have a $*$, then place subfiles $W_{n}^{f} , \forall n \in [N]$ in the cache of user $k$. By C1, we know that the symbol $``*"$ appears $Z$ times in each column, so each user stores $NZ$ subfiles. Since each subfile has size $L/F$ bits, the whole size of the cache is $NZL/F=ML$ bits, which satisfies users' cache constraint.
\subsubsection*{\textbf{Private Delivery Phase}} In this phase, each user will decide on their desired files. Let, user $k$ wants to retrieve file indexed by $d_k$. Then demand vector is $\dbf = {(d_k)}_{k \in [K]}$. Users don't want any of the servers to get any information about the demand vector. For that, users cooperatively generate queries for each server. Every user selects a random $N-1$ length vector independently and uniformly from ${[0:B-1]}^{N-1}$. User $k$ chooses $\Vbf^{k} = \{V_{0}^{k} , \ldots , V_{N-2}^{k}\}$, where $V_{n}^{k} \in [0:B-1]$.
After choosing these random vectors, each user generates $B$ vectors of length $N$ over $[0:B-1]$.
\begin{align}
  \Qbf_{b}^{k} &= \Bigg( V_{0}^{k} , V_{1}^{k} , \cdots , V_{d_k - 1}^{k} , {\Bigg( b - \sum_{n \in [0:N-2]} V_{n}^{k} \Bigg)}_{B} , \cdots , V_{N-2}^{k} \Bigg) \\
  &= \Big( Q_{b,0}^{k} , Q_{b,1}^{k} , \cdots , Q_{b,d_k - 1}^{k} , Q_{b,d_k}^{k} , \cdots , Q_{b,N-1}^{k} \Big)
\end{align}
where $Q_{b,n}^{k} \in [0:B-1]$. Note that
\[
  \Qbf_{b}^{k} \in \Qcal_{b} \triangleq \Bigg\{(q_{0} , q_{1} , \ldots , q_{N-1}) \in {[0:B-1]}^{N} \Big| {\Big(\sum_{n =0 }^{N-1}q_{n}\Big)}_{B} = b\Bigg\},
\]
i.e. the sum of the elements of $\Qbf_{b}^{k}$ is $b \bmod B$ for all $k \in [K]$.
Then, to server $b , \forall b \in [B]$ the following query is sent:
\begin{equation}
  \Qbf_{b} = \Big\{ \Qbf_{b}^{k} | k \in [K] \Big\}.
\end{equation}

After getting these queries, servers will generate answers based on the queries and the files they are storing. As stated earlier, every file is divided into $F$ subfiles, i.e. $W_n = W_n^f , \forall n \in [0:N-1]$, now servers will divide each subfiles into $B-1$ packets i.e.
\[
  W_n^f = \{ W_{n , b}^f : b \in [1:B-1]\} , \forall n \in [0:N-1] , \forall f \in [F].
\]
Each packet is of size $L / (B-1)F$ bits. For every $s \in [S]$, server $b$ computes
\begin{equation}
  A_{b , s}^{f , k} = W_{0 , Q_{b , 0}^{k}}^{f} + W_{1 , Q_{b , 1}^{k}}^{f} + \cdots + W_{N-1 , Q_{b , N-1}^{k}}^{f}~,~\forall (f , k) \in [F] \times [K] \mbox{ s.t. } \Pbf_{f , k} = s,
\end{equation}
where we define $W_{n , 0}^f = 0$ for any $n , f$. Defining
\begin{equation}
	X_{b,s} = \sum_{\{(f,k) \in [F]\times[K] : \Pbf_{f,k} = s\}} A_{b,s}^{f,k}
\end{equation}
server $b$, $b \in [B-1]$ transmits
\begin{equation}
	\Xbf_{b} = \Big( X_{b,1} , \ldots , X_{b,S} \Big).
\end{equation}
Note that if $\Qbf_{0}^{k} = \mathbf{0}, \forall k \in [K]$, then $X_{b,s} = 0$. Therefore if $\Qbf_{0}^{k} = \mathbf{0}, \forall k \in [K]$ then server $0$ won't transmit anything. Whereas if $\Qbf_{0}^{k} \neq \mathbf{0},$ for any $k \in [K]$ then server $0$ transmits
\begin{equation}
	\Xbf_{0} = \Big( X_{0,1} , \ldots , X_{0,S} \Big).
\end{equation}

\subsubsection*{Decoding}\label{sec:subsub:general_decoding}
Consider user $\kpr \in [K]$ and a subfile index $\fpr \in [F]$. If $\Pbf_{\fpr , \kpr} = *$ then user $\kpr$ recover subfile $W_{d_{\kpr}}^{\fpr}$ from the cache. But if $\Pbf_{\fpr , \kpr} = s$ for some $s \in [S]$ then the user decodes the desired subfile from the transmissons as follows:

For some $b \in [0:B-1]$ consider
\begin{align}
  X_{b,s} &= \sum_{\{(f,k) \in [F]\times[K] : \Pbf_{f,k} = s\}} A_{b,s}^{f,k} \\
  &= A_{b,s}^{\fpr , \kpr} + \sum_{\{(f,k) \in [F]\times[K] : \Pbf_{f,k} = s\} \setminus \{(\fpr , \kpr)\}} A_{b,s}^{f,k}.
\end{align}
Now, if $\Pbf_{\fpr , \kpr} = s$ and $\Pbf_{f , k} = s$ for some $(f , k) \in [F] \times [K]$ then C2 guarantees that $\Pbf_{f , \kpr} = *$ and therefore user $\kpr$ can access subfiles $W_{n}^{f}, \forall n \in [0:N-1]$ from its cache. As
\[
  A_{b,s}^{f,k} = \sum_{n = 0}^{N-1}W_{n , Q_{b,n}^{k}}^{f}
\]
user $\kpr$ can construct $A_{b,s}^{f,k}$ from its cache. Therefore user $\kpr$ can recover $A_{b , s}^{\fpr , \kpr}$ form $X_{b , s}$, as
\begin{align*}
  A_{b , s}^{\fpr , \kpr} &= X_{b , s} - \sum_{\{(f,k) \in [F]\times[K] : \Pbf_{f,k} = s\} \setminus \{(\fpr , \kpr)\}} A_{b,s}^{f,k} \\ &= X_{b , s} - \sum_{\{(f,k) \in [F]\times[K] : \Pbf_{f,k} = s\} \setminus \{(\fpr , \kpr)\}} \sum_{n = 0}^{N-1}W_{n , Q_{b,n}^{k}}^{f}.
\end{align*}
Now, user $\kpr$ has $A_{b,s}^{\fpr , \kpr} , \forall b \in [0:B-1]$. Defining
\begin{align*}
  \overline{W}_{k}^{f} &\triangleq \sum_{n=0}^{d_{k} - 1} W_{n , V_{n}^{k}}^{f} + \sum_{n=d_{k}+1}^{N-2} W_{n , V_{n}^{k}}^{f}~~~,~~~\forall (k , f) \in [K] \times [F] , \\
  \overline{V}^{k} &\triangleq \sum_{n = 0}^{N-2} V_{n}^{k}~~~,~~~\forall k \in [K]
\end{align*}
and recalling that $\Qbf_{b}^{k} = \Bigg( V_{0}^{k} , V_{1}^{k} , \cdots , V_{d_k - 1}^{k} , {\Big( b - \overline{V}^{k}\Big)}_{B} , \cdots , V_{N-2}^{k} \Bigg)$, we have
\begin{align*}
	A_{b,s}^{\fpr,\kpr} &= \sum_{n = 0}^{N-1}W_{n , Q_{b,n}^{\kpr}}^{\fpr} \\
                      &= \sum_{n=0}^{d_{\kpr} - 1} W_{n , V_{n}^{\kpr}}^{\fpr} + W_{d_{\kpr} , {(b - \overline{V}^{\kpr})}_{B}}^{\fpr} + \sum_{n=d_{\kpr}+1}^{N-2} W_{n , V_{n}^{\kpr}}^{\fpr} \\
                      &= W_{d_{\kpr} , {(b - \overline{V}^{\kpr})}_{B}}^{\fpr} + \overline{W}_{\kpr}^{\fpr}.
\end{align*}
As $W_{n , 0}^{f} = 0, \forall n,f$ we have
\[
  A_{{(\overline{V}^{\kpr})}_{B} , s}^{\fpr , \kpr} = \overline{W}_{\kpr}^{\fpr}.
\]
For some $\bpr \in [B-1]$, user $\kpr$ recovers packet $W_{d_{\kpr} , \bpr}^{\fpr}$ form $A_{{(\bpr+\overline{V}^{\kpr})}_{B} , s}^{\fpr , \kpr}$ and $A_{{(\overline{V}^{\kpr})}_{B} , s}^{\fpr , \kpr}$ as
\begin{align*}
  &A_{{(\bpr+\overline{V}^{\kpr})}_{B} , s}^{\fpr , \kpr} - A_{{(\overline{V}^{\kpr})}_{B} , s}^{\fpr , \kpr} \\
  = &W_{d_{\kpr} , {({(\bpr+\overline{V}^{\kpr})}_{B} - \overline{V}^{\kpr})}_{B}}^{\fpr} + \overline{W}_{\kpr}^{\fpr} - \overline{W}_{\kpr}^{\fpr} \\
  = &W_{d_{\kpr} , \bpr}^{\fpr}.
\end{align*}
Hence, the user is able to recover all packets of the subfile $W_{d_{\kpr}}^{\fpr}$.

\subsubsection*{Rate}\label{sec:subsub:general_rate}
Except for user $0$, all other users are transmitting $\frac{SL}{F(B-1)}$ bits. Therefore, $R_{b} = \frac{S}{F(B-1)}, \forall b \in [B-1]$. Whereas the number of bits broadcasted by user $0$ depends on the realizations of the queries. Consider
\begin{align*}
	X_{0,s} &= \sum_{\{(f,k) \in [F]\times[K] : \Pbf_{f,k} = s\}} A_{0,s}^{f,k} \\
  &= \sum_{\{(f,k) \in [F]\times[K] : \Pbf_{f,k} = s\}} \sum_{n = 0}^{N-1}W_{n , Q_{0,n}^{k}}^{f}. \\
\end{align*}
If $Q_{0,n}^{k} = 0 , \forall n \in [0:N-1]$ and $\forall k \in \Kappa_{s}$ which can happen with probability $B^{-|\Kappa_{s}|(N-1)}$ then $X_{0,s} = 0$ and server $0$ won't broadcast $X_{0,s}$. But if $Q_{0,n}^{k} \neq 0$ for some $n \in [0:N-1]$ and for some $k \in \Kappa_{s}$, then server $0$ broadcasts $X_{0,s}$ of size $\frac{L}{F(B-1)}$ bits to the users. Let, $R_{0,s}L$ be the size of the transmission $X_{0,s}$. Then
\begin{align*}
	\Ebb\Big\{R_{0,s}L\Big\} &= \Bigg(1 - \frac{1}{B^{|\Kappa_{s}|(N-1)}}\Bigg)\frac{L}{F(B-1)} \\
                           &= \frac{B^{|\Kappa_{s}|(N-1)} - 1}{B^{|\Kappa_{s}|(N-1)}} \frac{L}{F(B-1)} \\
                           &= \frac{L}{F} \frac{1 + B + \ldots + B^{|\Kappa_{s}|(N-1) - 1}}{B^{|\Kappa_{s}|(N-1)}} \\
                           &= \frac{L}{F} \RPIRs{B}{|\Kappa_{s}| (N-1)} \\
  \implies \Ebb\Big\{R_{0,s}\Big\} &= \frac{1}{F} \RPIRs{B}{|\Kappa_{s}| (N-1)}. \\
\end{align*}
Therefore, the average size of the broadcast performed by server $0$ is
\begin{equation}
	\Ebb\Big\{ R_{0} L \Big\} = \frac{L}{F} \sum_{s \in [S]} \RPIRs{B}{|\Kappa_{s}| (N-1)}.
\end{equation}
And the average number of bits broadcasted by every server is
\begin{align*}
	RL &= \frac{L}{F} \sum_{s \in [S]} \RPIRs{B}{|\Kappa_{s}| (N-1)} + \frac{(B-1)SL}{F(B-1)} \\
  &= \frac{L}{F} \Bigg( S + \sum_{s \in [S]} \RPIRs{B}{|\Kappa_{s}| (N-1)} \Bigg).
\end{align*}
Therefore the rate is given by
\begin{equation}
	R = \frac{S}{F} \Bigg( 1 + \frac{1}{S}\sum_{s \in [S]} \RPIRs{B}{|\Kappa_{s}| (N-1)} \Bigg).
\end{equation}

\subsubsection*{Proof of Privacy}\label{sec:subsub:general_proof_privacy}
Consider $\Qbf_{b}$, the query sent to server $b$. For any $\qbf \in {\Qcal_{b}}^{K}$ and for any $\abf = (a_{1} \ldots a_{K}) \in {[0:N-1]}^{K}$ we show that
\[
  \Prob{\dbf = \abf | \Qbf_{b} = \qbf} = \Prob{\dbf = \abf}.
\]
Let $\qbf = (\qbf^{1} , \ldots , \qbf^{K})$ where $\qbf^{k} \in \Qcal_{b}$. Consider
\begin{align*}
  \Prob{\dbf = \abf | \Qbf_{b} = \qbf} &= \frac{\Prob{\Qbf_{b} = \qbf | \dbf = \abf} \Prob{\dbf = \abf}}{\Prob{\Qbf_{b} = \qbf}}, \mbox{ and } \\
  \Prob{\Qbf_{b} = \qbf | \dbf = \abf} &= \Prob{ \Qbf_{b}^{k} = \qbf^{k} , \forall k \in [K] | \dbf = \abf } \\
  &= \begin{aligned}[t]
       \Prob{&\Vbf^{k}(0:a_{k} - 1)=\qbf^{k}(0:a_{k}-1) ,\\
       &\Vbf^{k}(a_{k}:N-2)=\qbf^{k}(a_{k}+1:N-1) , \forall k \in [K]}.
     \end{aligned}
\end{align*}
As all $\Vbf_{k}, \forall k \in [K]$ are chosen independently we have
\begin{align*}
  \Prob{\Qbf_{b} = \qbf | \dbf = \abf} &= \frac{1}{{B}^{K(N-1)}} \\
  \implies \Prob{\dbf = \abf | \Qbf_{b} = \qbf} &= \frac{1}{N^{K}} = \Prob{\dbf = \abf}
\end{align*}
which proves that a query sent to any server is independent of the demand vector, and therefore none of the servers can get any information about the demand vector.

\subsubsection*{Subpacketization level and Upload cost}\label{sec:subsub:general_subpacket_upload}
Every file is divided into $F$ subfiles, and every subfile is further divided into $B - 1$ packets. So, the subpacketization level is $F(B-1)$. Every user is constructing one query for every server. Therefore there are $BK$ queries in total. Each query consists of $N$ integers out of which $N-1$ integers are chosen uniformly from $[0:B-1]$, and the ${N}^{th}$ integer is dependent on the other $N-1$. Therefore, the upload cost is $BK(N-1)\log_{2}B$ bits.

\section{Conclusion}
In this paper, we considered the problem of cache-aided multi-user private information retrieval. We considered dedicated cache setups where all users have a cache of equal size. The MuPIR strategy we proposed utilizes placement delivery arrays to specify placements and private deliveries by the servers. The subpacketization level of MuPIR schemes depends on the subpacketization level of the PDAs. Therefore, MuPIR schemes utilizing low subpacketization PDAs also have low subpacketization level as well as lower upload costs compared to already existing schemes e.g.~\cite{Ming21CaMuPIR}. We also proved the order optimality of the MuPIR schemes that are based on the PDAs corresponding to the MAN scheme. Then we analyzed MuPIR schemes based on PDAs described in~\cite{QYan17PDA}. The resulting MuPIR scheme has a marginally higher rate than the rate achieved in product design of~\cite{Ming21CaMuPIR}, but there is a significant improvement in subpacketization level and upload cost.

Although order optimal results are provided for the special case of MAN-based PDA, exact optimality results are open in terms of rate, subpacketization level and upload cost for cache-aided MuPIR problem. Furthermore, this paper didn't consider the case of multi-access caches, which generalizes dedicated cache setups. In a multi-access cache setup, the subpacketization level is lower than the dedicated cache setup for the same number of users. Still, more work is needed to further reduce the subpacketization and the upload cost for multi-access setups.

\section*{Acknowledgement}
This work was supported partly by the Science and Engineering Research Board (SERB) of Department of Science and Technology (DST), Government of India, through J.C. Bose National Fellowship to B. Sundar Rajan, and by the Ministry of Human Resource Development (MHRD), Government of India, through Prime Minister’s Research Fellowship (PMRF) to Kanishak Vaidya.

\clearpage
\end{document}